\documentclass[]{llncs}

\usepackage{amsmath}
\usepackage{amssymb}

\usepackage[utf8]{inputenc}
\usepackage[T1]{fontenc} 

\newtheorem{numremark}[theorem]{Remark}

\newcommand{\myquot}[1]{``#1''}

\newcommand{\proj}[1]{\pi_I(#1)}
\newcommand{\set}[1]{\{#1\}}

\newcommand{\sepp}{\#}
\newcommand{\blank}{*}

\newcommand{\nats}{\mathbb{N}}

\newcommand{\size}[1]{|#1|}

\renewcommand{\epsilon}{\varepsilon}
\renewcommand{\phi}{\varphi}

\newcommand{\U}{\mathrm{U}}

\newcommand{\aut}{\mathcal{A}}
\newcommand{\auttrack}{\mathcal{T}}

\newcommand{\ops}{\mathrm{Ops}}
\newcommand{\inc}[1]{#1 := #1 + 1}
\newcommand{\reset}[1]{#1 := 0}
\newcommand{\maxx}[3]{#1 := \max(#2,#3)}
\newcommand{\val}{\nu}

\newcommand{\lbls}{\Lambda}

\newcommand{\flatten}{\mathrm{flat}}
\newcommand{\run}{\rho}

\newcommand{\delaygame}[1]{\Gamma\!_{f}(#1)}

\newcommand{\SigmaO}{\Sigma_O}
\newcommand{\SigmaI}{\Sigma_I}
\newcommand{\stratO}{\tau_O}
\newcommand{\stratI}{\tau_I}
\newcommand{\p}{p}

\newcommand{\game}{\mathcal{G}}

\newcommand{\idx}[1]{\mathrm{idx}(#1)}
\newcommand{\eqclass}[1]{[#1]}

\newcommand{\eqop}{\equiv_{\mathrm{ops}}}

\newcommand{\eqoppar}[1]{\eqop^{#1}}

\newcommand{\eqwordpar}[1]{\equiv_\aut^{#1}}
\newcommand{\eqclasswordpar}[2]{\eqclass{#1}_{\eqwordpar{#2}}}
\newcommand{\quotientwordpar}[1]{\hspace{-0pt}\slash\hspace{-0pt}\eqwordpar{#1}}

\newcommand{\eqwordparpro}[1]{=_\aut^{#1}}
\newcommand{\eqclasswordparpro}[2]{\eqclass{#1}_{\eqwordparpro{#2}}}

\newcommand{\exptime}{\textsc{ExpTime}}
\newcommand{\pspace}{\textsc{PSpace}}

\title{Unbounded Lookahead in WMSO+U Games\thanks{Supported by the project ``TriCS'' (ZI~1516/1-1) of the German Research Foundation (DFG).}}

\author{Martin Zimmermann}

\institute{Reactive Systems Group, Saarland University, 66123 Saarbr\"ucken, Germany\\
\email{zimmermann@react.uni-saarland.de}}

\begin{document}

\maketitle

\begin{abstract}
\noindent Delay games are two-player games of infinite duration in which one player may delay her moves to obtain a lookahead on her opponent's moves. We consider delay games with winning conditions expressed in weak monadic second order logic with the unbounding quantifier (WMSO$+$U), which is able to express (un)boundedness properties. It is decidable whether the delaying player is able to win such a game with bounded lookahead, i.e., if she only skips a finite number of moves.

However, bounded lookahead is not always sufficient: we present a game that can be won with unbounded lookahead, but not with bounded lookahead. Then, we consider WMSO$+$U delay games with unbounded lookahead and show that the exact evolution of the lookahead is irrelevant: the winner is always the same, as long as the initial lookahead is large enough and the lookahead tends to infinity. 

\end{abstract}

\section{Introduction}
\label{sec_intro}
Many of today's problems in computer science are no longer concerned with
programs that transform data and then terminate, but with non-terminating
reactive systems which have to interact with a possibly antagonistic
environment for an unbounded amount of time. The framework of infinite
two-player games is a powerful and flexible tool to verify and synthesize such
systems. The seminal theorem of Büchi and Landweber~\cite{BuechiLandweber69}
states that the winner of an infinite game on a finite arena with an $\omega$-regular
winning condition can be determined and a corresponding finite-state winning
strategy can be constructed effectively.

Ever since, this result was extended along different dimensions, e.g., the
number of players, the type of arena, the type of winning condition, the type
of interaction between the players (alternation or concurrency), zero-sum or
non-zero-sum, and complete or incomplete information. In this work, we consider
two of these dimensions, namely more expressive winning conditions and the
possibility for one player to delay her moves.

\paragraph{WMSO$+$U} Recall that the $\omega$-regular languages are exactly
those that are definable in monadic second order logic (MSO)~\cite{Buechi62}. Recently, Bojańczyk has started a
program~\cite{Bojanczyk04,Bojanczyk11,Bojanczyk14,BojanczykColcombet06,BojanczykGMS14,BPT15,BojanczykTorunczyk09,BojanczykTorunczyk12,HummelSkrzypczak12} 
investigating the logic MSO$+$U, MSO extended with the unbounding
quantifier~$\U$. A formula~$\U X \phi(X)$ is satisfied, if there are
arbitrarily large \emph{finite} sets~$X$ such that $\phi(X)$ holds. MSO$+$U is
able to express all $\omega$-regular languages as well as non-regular ones such as $L = \set{ a^{n_0} b a^{n_1} b
a^{n_2} b \cdots \mid \limsup\nolimits_i n_i = \infty }$ .
Decidability of MSO$+$U remained an open problem until recently: satisfiability of MSO$+$U on infinite words is undecidable~\cite{BPT15}. 

Even before this undecidability result was shown, much attention was being paid to fragments of the logic obtained by restricting the power of the second-order quantifiers. In particular, considering
weak\footnote{Here, the second-order quantifiers are restricted to finite
sets.} MSO with the unbounding quantifier (denoted by prepending a W) turned
out to be promising: WMSO$+$U on infinite words~\cite{Bojanczyk11} and on
infinite trees~\cite{BojanczykTorunczyk12} and WMSO$+$U with the path
quantifier (WMSO$+$UP) on infinite trees~\cite{Bojanczyk14} have equivalent automata models with decidable emptiness. Hence, these
logics are decidable.

For WMSO$+$U on infinite words, these automata are called max-automata,
deterministic automata with counters whose acceptance conditions are a boolean
combination of conditions~\myquot{counter~$c$ is bounded during the run}. While
processing the input, a counter may be incremented, reset to zero, or the
maximum of two counters may be assigned to it (hence the name max-automata). In
this work, we continue the investigation of delay games with winning conditions given by max-automata, so-called max-regular conditions. 

\paragraph{Delay Games} In such a delay game, one of the players can postpone her moves for some time,
thereby obtaining a lookahead on her opponent's moves. This allows her to
win some games which she loses without lookahead, e.g., if her first move
depends on the third move of her opponent. Nevertheless, there are
winning conditions that cannot be won with any finite lookahead, e.g., if her
first move depends on every move of her opponent. Delay
arises naturally when transmission of data in networks or components
with buffers are modeled.

From a more theoretical point of view,
uniformization of relations by continuous functions~\cite{Thomas11,DBLP:conf/rex/ThomasL93,trakhtenbrot1973finite} can be expressed and
analyzed using delay games. We consider games in which two players pick letters from alphabets~$\SigmaI$
and $\SigmaO$, respectively, thereby producing $\alpha \in \SigmaI^\omega$
and $\beta \in \SigmaO^\omega$. Thus, a strategy for the second player induces a mapping
$\tau\colon \SigmaI^\omega\rightarrow\SigmaO^\omega$. It is winning for the
second player if $(\alpha,\tau(\alpha))$ is contained in the winning
condition~$L\subseteq \SigmaI^\omega\times \SigmaO^\omega$ for every $\alpha$. Then, we say that $\tau$ uniformizes~$L$. 

In the classical setting of infinite games, in
which the players pick letters in alternation, the $n$-th letter of
$\tau(\alpha)$ depends only on the first $n$ letters of $\alpha$, i.e., $\tau$ satisfies a very strong notion of continuity. A strategy
with bounded lookahead, i.e., only finitely many
moves are postponed, induces a Lipschitz-continuous function $\tau$ (in the
Cantor topology on $\Sigma^\omega$) and a strategy with arbitrary lookahead
induces a continuous function (or equivalently, a uniformly continuous
function, as $\Sigma^\omega$ is compact).

Hosch and Landweber proved that it is decidable whether a game with $\omega$-regular
winning condition can be won with bounded lookahead~\cite{HoschLandweber72}. This result was improved by
Holtmann, Kaiser, and Thomas who showed that if a
player wins a game with arbitrary lookahead, then she wins already with
doubly-exponential bounded lookahead, and gave a streamlined decidability
proof yielding an algorithm with doubly-exponential running time~\cite{HoltmannKaiserThomas12}. Again, these
results were improved by giving an exponential upper bound on the necessary
lookahead and showing $\exptime$-completeness of the solution problem~\cite{KleinZimmermann15}. Going beyond $\omega$-regular winning
conditions by considering context-free conditions leads to undecidability and
non-elementary lower bounds on the necessary lookahead, even for very weak
fragments~\cite{FridmanLoedingZimmermann11}. In contrast, studying delay games with WMSO$+$U turned out to be more fruitful~\cite{Zimmermann15}: the winner of such a game w.r.t.\ bounded lookahead is decidable, i.e., the Hosch-Landweber Theorem holds for max-regular conditions, too. 

Stated in terms of uniformization, Hosch and Landweber proved decidability of the
uniformization problem for $\omega$-regular relations by Lipschitz-continuous functions
and Holtmann et al.\ proved the equivalence of the existence of a
continuous uniformization function and the existence of a Lipschitz-continuous
uniformization function for $\omega$-regular relations. 

In another line of work, Carayol and Löding considered the case of finite words~\cite{CarayolLoeding12}, and Löding and Winter~\cite{LoedingWinter14} considered the case of finite trees, which are both decidable. However, the nonexistence of MSO-definable choice functions on the infinite
binary tree~\cite{CarayolLoeding07,GS83} implies that uniformization fails for such trees.

Another application of delay games concerns the existence of Wadge reductions between max-regular languages~\cite{CabessaDFM09}, which can be expressed as a max-regular delay game.

\paragraph{Our Contribution} 

Here, we continue the investigation of delay games with max-regular winning conditions, which was started by proving the analogue of the Hosch-Landweber Theorem~\cite{Zimmermann15}: the winner of a game w.r.t.\ bounded lookahead is decidable. In particular, we are interested in the analogue of the Holtmann-Kaiser-Thomas Theorem (is bounded lookahead sufficient?). Not surprisingly, our first result (which was already announced, but not proved, in~\cite{Zimmermann15}) shows that this does not hold: unbounded lookahead is more powerful when it comes to unboundedness conditions. 

We complement this by showing that the ability of Player~$O$ to win a max-regular delay game does not depend on the growth rate, only on the fact that it grows without bound and a sufficiently large initial lookahead. This is, to the best of our knowledge, the first such result and should be contrasted with the case of $\omega$-context-free winning conditions, for which a non-elementary growth rate might be necessary for Player~$O$ to win~\cite{FridmanLoedingZimmermann11}.

As the analogue of the Holtmann-Kaiser-Thomas Theorem fails, determining the winner of max-regular delay games with respect to arbitrary delay functions does not coincide with determining the winner with respect to bounded delay functions. Hence, we investigate the former problem: we give lower bounds on the complexity and discuss some obstacles one encounters when trying the extend the decidability proof for the bounded case and the undecidability proof for MSO$+$U satisfiability. 

\section{Preliminaries}
\label{sec_prel}
The set of non-negative integers is denoted by $\nats$. An alphabet $\Sigma$ is a non-empty finite set of letters, and $\Sigma^{*}$ ($\Sigma^n$, $\Sigma^{\omega}$) denotes the set of finite words (words of length $n$, infinite words) over $\Sigma$. The empty word is denoted by $\varepsilon$, the length of a finite word~$w$ by $|w|$. For $w\in \Sigma^{*}\cup\Sigma^{\omega}$ we write $w(n)$ for the $n$-th letter  of $w$. Given two infinite words $\alpha \in \SigmaI^\omega$ and $\beta  \in \SigmaO^\omega$ we write ${ \alpha \choose \beta}$ for the word ${\alpha(0) \choose \beta(0) } {\alpha(1) \choose \beta(1) }  {\alpha(2) \choose \beta(2) } \cdots  \in (\SigmaI \times \SigmaO)^\omega$. Analogously, we write ${x \choose y}$ for finite words $x$ and $y$, provided they are of equal length. Finally, the index of an equivalence relation~$\equiv$, i.e., the number of its equivalence classes, is denoted by $\idx{\equiv}$.

\paragraph*{Max-Automata}
\label{subsec_automata}
Given a finite set~$C$ of counters storing non-negative integers,
\[
\ops(C) = \set{\inc{c}, \reset{c},  \maxx{c}{c_0}{c_1} \mid c,c_0, c_1 \in C}
\]
is the set of counter operations over $C$. A counter valuation over $C$ is a mapping~$\val\colon C \rightarrow \nats$. By $\val \pi$ we denote the counter valuation that is obtained by applying a finite sequence~$\pi \in \ops(C)^*$ of counter operations to $\val$, which is defined as implied by the operations' names.

A max-automaton~$\aut = (Q, C, \Sigma, q_I, \delta, \ell, \phi)$ consists of a finite set~$Q$ of states, a finite set~$C$ of counters, an input alphabet~$\Sigma$, an initial state~$q_I$, a (deterministic and complete) transition function~$\delta \colon Q \times \Sigma \rightarrow Q$, a transition labeling\footnote{Here, and later whenever convenient, we treat $\delta$ as relation~$\delta \subseteq Q \times \Sigma \times Q$.}~$\ell \colon \delta \rightarrow \ops(C)^*$ which labels each transition by a (possibly empty) sequence of counter operations, and an acceptance condition~$\phi$, which is a boolean formula over~$C$.

A run of $\aut$ on $\alpha \in \Sigma^\omega$ is an infinite sequence
\begin{equation}
\label{eq_run}
\rho = (q_0, \alpha(0), q_1) \,  (q_1, \alpha(1), q_2) \,  (q_2, \alpha(2), q_3) \cdots  \in \delta^\omega
\end{equation}
with $q_0 = q_I$. Runs on finite words are defined analogously, i.e., \[(q_0, \alpha(0), q_1) \cdots  (q_{n-1}, \alpha(n-1), q_n)\] is the run of $\aut$ on $\alpha(0) \cdots \alpha(n-1)$ starting in $q_0$. We say that this run ends with $q_n$. As $\delta$ is deterministic, $\aut$ has a unique run on every finite or infinite word.

Let $\rho$ be as in~(\ref{eq_run}) and define $\pi_n = \ell(q_n, \alpha(n), q_{n+1})$, i.e., $\pi_n$ is the label of the $n$-th transition of $\rho$. Given an initial counter valuation~$\val$ and a counter $c \in C$, we define the sequence~$ \rho_c = \val(c)\, , \, \val\pi_0(c)\, ,\,  \val\pi_0 \pi_1(c)\,  ,\,  \val\pi_0\pi_1\pi_2(c)\,,  \ldots  $
of counter values of $c$ reached on the run after applying \emph{all} operations of a transition label.
The run~$\rho$ of $\aut$ on $\alpha$ is accepting, if the acceptance condition~$\phi$ is satisfied by the variable valuation that maps a counter~$c$ to true if and only if $\limsup \rho_c$ is finite. Thus, $\phi$ can intuitively be understood as a boolean combination of conditions~\myquot{$\limsup \rho_c < \infty$}. Note that the limit superior of $\rho_c$ is independent of the initial valuation used to define $\rho_c$, which is the reason it is not part of the description of $\aut$. We denote the language accepted by $\aut$ by $L(\aut)$ and say that it is max-regular.

A parity condition (say min-parity) can be expressed in this framework using a counter for each color that is incremented every time this color is visited and employing the acceptance condition to check that the smallest color whose associated counter is unbounded, is even. Hence, the class of $\omega$-regular languages is contained in the class of max-regular languages.

\paragraph*{Delay Games}
\label{subsec_games}
A delay function is a mapping $f\colon\nats\rightarrow \nats \setminus \set{0}$, which is said to be bounded, if $f(i) = 1$ for almost all $i$. Otherwise, $f$ is unbounded. A special case of the bounded delay functions are the constant ones: delay functions~$f$ with $f(i) = 1$ for every $i > 0$.

Fix an input alphabet~$\SigmaI$ and an output alphabet~$\SigmaO$. Given a delay function~$f$ and an $\omega$-language $L\subseteq \left(\SigmaI\times\SigmaO\right)^\omega$, the game $\delaygame{L}$ is played by two players (Player~$I$ and Player~$O$) in rounds $i=0,1,2,\ldots$ as follows: in round $i$, Player~$I$ picks a word $u_i\in\SigmaI^{f(i)}$, then Player~$O$ picks one letter $v_i\in\SigmaO$. We refer to the sequence $(u_0,v_0),(u_1,v_1),(u_2,v_2),\ldots$ as a play of $\delaygame{L}$. Player~$O$ wins the play if and only if the outcome~${u_0u_1u_2\cdots \choose v_0v_1v_2\cdots}$ is in $L$, otherwise Player~$I$ wins.

Given a delay function $f$, a strategy for Player~$I$ is a mapping $\stratI\colon \SigmaO^*\rightarrow \SigmaI^*$ such that $|\stratI(w)|=f(|w|)$, and a strategy for Player~$O$ is a mapping $\stratO\colon \SigmaI^*\rightarrow \SigmaO$. Consider a play $(u_0,v_0),(u_1,v_1),(u_2,v_2),\ldots$ of $\delaygame{L}$. Such a play is consistent with $\stratI$, if $u_{i}=\stratI(v_0\cdots v_{i-1})$ for every $i$; it is consistent with $\stratO$, if $v_i=\stratO(u_0\cdots u_i)$ for every $i$. A strategy $\tau$ for Player~$\p$ is winning for her, if every play that is consistent with $\tau$ is won by Player~$\p$. In this case, we say Player~$p$ wins $\delaygame{L}$. 

Given a max-automaton~$\aut$, we want to determine whether Player~$O$ has a winning strategy for $\delaygame{L(\aut)}$ for some $f$, and, if yes, what kind of $f$ is sufficient to win. Note that due to monotonicity, Player~$O$ wins a delay game for an arbitrary winning condition w.r.t.\ a bounded delay function if and only if she wins the game w.r.t.\ a constant delay function. 

\section{Bounded Lookahead is not Always Sufficient}
\label{sec_bounded}
The winner of a max-regular delay game w.r.t.\ constant delay functions can be determined effectively by a reduction to delay-free games with max-regular winning conditions, i.e., the following problem is decidable~\cite{Zimmermann15}: given a max-automaton~$\aut$, does Player~$O$ win $\delaygame{L(\aut)}$ for some constant (equivalently, bounded) delay function~$f$? However, in this section, we show that bounded and thus constant lookahead does not suffice to win every delay game that Player~$O$ can win with arbitrary lookahead.

\begin{theorem}
\label{thm_unboundednecessary}
There is a max-regular language $L$ such that Player~$O$ wins $\delaygame{L}$ for every unbounded $f$, but not for any bounded $f$. 	
\end{theorem}

\begin{proof}
Let $\SigmaI = \set{0,1, \sepp}$ and $\SigmaO = \set{0,1,\blank}$. An input block is a word $\sepp w $ with $w \in \set{0,1}^+$. An output block is a word~$
{\sepp \choose \alpha(n)} 
{\alpha(1) \choose \blank} 
{\alpha(2) \choose \blank} 
\cdots
{\alpha(n-1) \choose \blank} 
{\alpha(n) \choose \alpha(n)}  \in (\SigmaI \times \SigmaO)^+$
with $\alpha(j) \in \set{0,1}$ for all $j$ in the range $1 \le j \le n$. Note that the first and last letter in an output block are the only ones whose second component is not an $\blank$, and that these letters have to be equal to the first component of the block's last letter. Also, note that neither an input block nor an output block has to be maximal in the sense that it has to end with a $\sepp$ (in the first component). Every input block of length $n$ can be extended to an output block of length $n$ and  projecting an output block to its first components yields an input~block.

Let  $L \subseteq  (\SigmaI \times \SigmaO)^\omega$ be the language of words ${\alpha \choose \beta}$ satisfying the following property: if $\alpha$ contains infinitely many $\sepp$ and arbitrarily long input blocks, then ${\alpha \choose \beta}$
contains arbitrarily long output blocks.
It is easy to come up with a WMSO$+$U formula defining $L$ by formalizing the definitions of input and output blocks in first-order logic. 

Now, consider $L$ as winning condition for a delay game. Intuitively, Player~$O$ has to specify arbitrarily long output blocks, provided Player~$I$ produces arbitrarily long input blocks. The challenge for Player~$O$ is that she has to specify at the beginning of every output block whether she ends the block in a position where Player~$I$ has picked a $0$ or a $1$. 


First, consider $\delaygame{L}$ for an unbounded delay function~$f$. The following strategy is winning for Player~$O$: whenever she has to pick $\beta(i)$ at a position where Player~$I$ picked $\alpha(i) = \sepp$, she picks the last letter of the longest input block in the lookahead that starts with the current $\sepp$. Then, she completes the output block by picking $\blank$ until the end of the input block, where she copies $\beta(i)$, which completes the output block. At every other position, she picks an arbitrary letter. Now, consider a play consistent with this strategy: if Player~$I$ picks infinitely many $\sepp$ and arbitrarily large input blocks, then Player~$O$ will see arbitrarily large input blocks in her lookahead, i.e., her strategy picks arbitrarily large output blocks. Thus, the strategy is indeed winning.

It remains to show that Player~$I$ wins $\delaygame{L}$ for every bounded delay function~$f$. Fix such a function and define $\ell = \sum_{i\colon f(i)>0}f(i)-1$, i.e., $\ell$ is the maximal lookahead size Player~$O$ will achieve. Player~$I$ produces longer and longer input blocks of the following form: he starts picking $\sepp$ followed by $0$'s until Player~$O$ has picked an answer at the position of the last $\sepp$. If she picked a $0$, then Player~$I$ finishes the input block by picking $1$'s; if she picked a $1$ (or an $\blank$), then he finishes the input block by picking $0$'s. Thus, the length of every output block is at most $\ell$, since Player~$O$ has to determine the answer to every $\sepp$ after seeing at most the next $\ell$ letters picked by Player~$I$. Thus, Player~$I$ picks infinitely many $\sepp$ and arbitrarily long input blocks, while the length of the output blocks is bounded. Hence, the strategy is winning for Player~$I$.

\end{proof}

\section{Any Unbounded Lookahead is Sufficient}
\label{sec_unbounded}
In this section, we complement the result of the previous section, showing that bounded lookahead is not always sufficient for max-regular delay games, by showing that any unbounded lookahead is sufficient for Player~$O$, provided some lookahead allows her to win at all. If Player~$O$ wins a game with respect to some delay function~$f$, then she also wins with respect to every $f'$ that grants her at every round larger lookahead\footnote{This holds for every winning condition, not only max-regular ones.}. The hard part of the proof is to show that she also wins for \emph{smaller} functions~$f'$ that grant her less lookahead. 

To this end, in Subsection~\ref{subsec_equiv}, we introduce equivalence relations that capture the behavior of a max-automaton up to a certain precision. Then, in Subsection~\ref{subsec_equivgame}, we define an infinite-state game~$\game$ based on these equivalence relations. Intuitively, the players' pick equivalence classes and Player~$I$ is in charge of increasing the precision of the approximation of the automaton's behavior, i.e., there is no explicit delay function in the definition of $\game$. This game allows to prove that smaller, but unbounded, lookahead is also sufficient. 

\subsection{Equivalence Relations for Max-Automata}
\label{subsec_equiv}
Fix $\aut = (Q, C, \Sigma, q_I, \delta, \ell, \phi)$. We generalize notions introduced in \cite{Bojanczyk11} and \cite{Zimmermann15} to define equivalences over sequences of counter operations and over words over $\Sigma$ to capture the behavior of $\aut$ up to a given precision. To this end, we need to introduce some notation to deal with runs of $\aut$. Given a state~$q$ and $w \in \Sigma^* \cup \Sigma^\omega$, let $\run(q, w)$ be the run of $\aut$ on $w$ starting in $q$. If $w$ is finite, then $\delta^*(q, w)$ denotes the state $\run(q,w)$ ends with. The transition profile of $w \in \Sigma^*$ is the mapping~$q \mapsto \delta^*(q,w)$. 

Now, we define inductively what it means for a sequence~$\pi \in \ops(C)^*$ to transfer a counter~$c$ to a counter~$d$. The empty sequence and the operation~$\inc{c}$ transfer every counter to itself. The operation~$\reset{c}$ transfers every counter~$c' \neq c$ to itself and the operation~$\maxx{c}{c_0}{c_1}$ transfers every counter but $c$ to itself and transfers $c_0$ and $c_1$ to $c$. Finally, if $\pi_0$ transfers $c$ to $e$ and $\pi_1$ transfers $e$ to $d$, then $\pi_0\pi_1$ transfers $c$ to $d$. If $\pi$ transfers $c$ to $d$, then we have $\val \pi (d) \ge \val(c)$ for every counter valuation~$\val$, i.e., the value of $d$ after executing $\pi$ is larger or equal to the value of $c$ before executing $\pi$, independently of the initial counter values.

Furthermore, a sequence of counter operations~$\pi$ transfers $c$ to $d$ with $m \ge 0$ increments, if there are counters~$e_1, \ldots, e_m$ and a decomposition
\[\pi = \pi_0\,(\inc{e_1}) \,\pi_1 \,(\inc{e_2}) \,\pi_2\,\cdots\, \pi_{m-1} \,(\inc{e_m}) \,\pi_{m}\] of $\pi$ such that $\pi_0$ transfers $c$ to $e_1$, $\pi_{j}$ transfers $e_j$ to $e_{j+1}$ for every $j$ in the range~$1 \le j < m$, and $\pi_m$ transfers $e_m$ to $d$. If $\pi$ transfers $c$ to $d$ with $m$ increments, then we have $\val \pi (d) \ge \val(c) + m$ for every counter valuation~$\val$. Also note that if $\pi$ transfers $c$ to $d$ with $m>0$ increments, then it also transfers $c$ to $d$ with $m'$ increments for every $m' \le m$. Finally, we say that $\pi$ is a $c$-trace of length~$m$, if there is a counter~$c'$ such that $\pi$ transfers $c'$ to $c$ with $m$ increments. Thus, if $\pi$ is a $c$-trace of length~$m$, then $\val \pi (c) \ge m$ for every valuation~$\val$. 

As only counter values reached after executing all counter operations of a transition label are considered in the semantics of max-automata, we treat $\lbls = \set{\ell(q,a,q') \mid (q, a, q') \in \delta}$ as an alphabet. Every word $\lambda \in \lbls^*$ can be flattened to a word in $\ops(C)^*$, which is denoted by $\flatten(\lambda)$. However, infixes, prefixes, or suffixes of $\lambda$ are defined with respect to the alphabet $\lbls$. We define $\ell(q, w) \in \lbls^*$ to be the sequence of elements in $\lbls$ labeling the run~$\run(q, w)$.

Let $\rho$ be a finite run of $\aut$ and let $\pi \in \ops(C)^*$. We say that $\rho$ ends with $\pi$, if $\pi$ is a suffix of $\flatten(\ell(\rho))$. A finite or infinite run contains $\pi$, if it has a prefix that ends with $\pi$.

\begin{lemma}[\cite{Bojanczyk11}]
\label{lemma_ctraceswitnessunboundedness}
Let $\rho$ be a run of $\aut$ and $c$ a counter. Then, $\limsup \rho_c = \infty$ if and only if $\rho$ contains arbitrarily long $c$-traces.
\end{lemma}

We use the notions of transfer (with increment) to define the equivalence relations that capture $\aut$'s behavior. Fix some $m\ge 0$. We say that $\lambda, \lambda' \in \lbls^*$ are $m$-equivalent, denoted by $\lambda \eqoppar{m} \lambda'$, if for all counters~$c$ and $d$ and for all $m'$ in the range $0 \le m' \le m$: 
\begin{enumerate}

	\item $\lambda$ has an infix whose flattening has a suffix that is a $c$-trace of length $m'$ if and only if $\lambda'$ has an infix whose flattening has a suffix that is a $c$-trace of length $m'$, 

	\item the flattening of $\lambda$ has a suffix that is a $c$-trace of length~$m'$ if and only if the flattening of $\lambda'$ has a suffix that is a $c$-trace of length~$m'$,

	\item the flattening of $\lambda$ transfers $c$ to $d$ with $m'$ increments if and only if the flattening of $\lambda'$ transfers $c$ to $d$ with $m'$ increments, and

	\item $\lambda$ has a prefix whose flattening transfers $c$ to $d$ with $m'$ increments if and only if $\lambda'$ has a prefix whose flattening transfers $c$ to $d$ with $m'$ increments.

\end{enumerate}
Using this, we define two words $x,x' \in \Sigma^*$ to be $m$-equivalent, denoted by $x \eqwordpar{m} x'$, if they have the same transition profile and if $\ell(q, x)\eqoppar{m} \ell(q, x')$ for all states $q$. 

Recall that a congruence is an equivalence relation~$\equiv$ over $\Sigma^*$ such that $x \equiv y$ implies $xz \equiv yz$ for every $z \in \Sigma^*$.

\begin{lemma}
\label{lemma_eqproperties}
Let $\aut$ be a max-automaton with $n$ states and $k$ counters and let $m \in \nats$.
\begin{enumerate}
	
	\item $\lambda \eqoppar{m} \lambda'$ implies $\lambda \eqoppar{m'} \lambda'$ for every $m' \le m$.
 	
 	\item $x \eqwordpar{m} x'$ implies $x \eqwordpar{m'} x'$ for every $m' \le m$.
 	
 	\item\label{lemma_eqproperties_eqoppar_cong}
 	$\eqoppar{m}$ is a congruence.
 	
 	\item\label{lemma_eqproperties_eqwordpar_cong}
 	$\eqwordpar{m}$ is a congruence.
	
	\item\label{lemma_eqproperties_eqoppar_index}
	The index of $\eqoppar{m}$ is at most $2^{2(k^2+k)\log(m+2)}$.
	
	\item\label{lemma_eqproperties_eqwordpar_index}
	The index of $\eqwordpar{m}$ is at most $2^{n(\log(n) + 2(k^2+k)\log(m+2))}$.
\end{enumerate}
\end{lemma}

\begin{proof}
The first two items follow trivially from the definition of $\eqoppar{m}$. Thus, we only consider the latter four items.

\ref{lemma_eqproperties_eqoppar_cong}. Let $\lambda \eqoppar{m} \lambda'$ and let $\pi \in \Lambda$ (note that we treat $\pi$ as a letter from $\lambda$, although it is also a sequence of counter operations). We show $\lambda \pi \eqoppar{m} \lambda' \pi$. An inductive application proves that $\eqoppar{m}$ is an equivalence. 

First, assume $\lambda \pi$ has an infix~$\lambda_0$ whose flattening has a suffix~$\pi_0$ that is a $c$-trace of length $m'$ for some $m' \le m$. If $\lambda_0$ is an infix of $\lambda$, then $\lambda \eqoppar{m} \lambda'$ implies that $\lambda'$ has an infix with the same property. The other trivial case is when $\lambda_0$ is equal to $\pi$. Thus, it remains to consider the case where $\lambda_0$ is a suffix of $\lambda \pi$ of length at least two (recall that we treat $\pi$ as one letter, i.e., $\lambda_0$ contains at least one letter from $\lambda$). Thus, $\lambda_0$ can be decomposed into two parts, one that is a $c'$-trace of length~$m_0$ and is a suffix of the flattening of $\lambda$, and another one that is equal to $\pi$ (treated as a sequence of counter operations now), which transfers $c'$ to $c$ with $m_1$ increments. Furthermore, we have $m_0 + m_1 = m' \le m$. 

Due to $\lambda \eqoppar{m} \lambda'$, we conclude that the flattening of $\lambda'$ has a suffix that is a $c'$-trace of length~$m_0$. Combining this suffix with $\pi$, we obtain a suffix of the flattening of $\lambda'\pi$ that is a $c$-trace of length~$m'$. This is also an infix of $\lambda'\pi$ whose flattening has a suffix that is a $c$-trace of length $m'$. 

The argument where $\lambda' \pi$ has such an infix is symmetric and the reasoning for the other three properties in the definition of $\eqoppar{m}$ is analogous. 

\ref{lemma_eqproperties_eqwordpar_cong}. Having the same transition profile is a congruence, since $\delta^*(q, xz) = \delta^*(\delta^*(q, x), z)$. This, and $\eqoppar{m}$ being a congruence imply that $\eqwordpar{m}$ is a congruence as well.

\ref{lemma_eqproperties_eqoppar_index}. An equivalence class of $\eqoppar{m}$ is uniquely characterized by the following properties:
\begin{itemize}
	\item for every counter~$c$, whether its elements have an infix whose flattening has a suffix that is a $c$-trace, and if yes by the largest~$m' \le m$ such that the length of such a $c$-trace is $m'$.
	\item For every counter~$c$, whether its elements have a suffix that is a $c$-trace, and if yes by the largest $m' \le m$ such that the length of such a trace is $m'$.
	\item For every pair $(c,d)$ of counters, whether the flattenings of its elements transfer $c$ to $d$, and if yes by the largest $m' \le m$ such that the transfer has $m'$ increments.
	\item For every pair~$(c,d)$ of counters, whether its elements have a prefix whose flattening transfers $c$ to $d$, and if yes by the largest $m' \le m$ such that the transfer has $m'$ increments.
\end{itemize}
Thus, an equivalence class is induced by two mappings from $C$ to $\set{\bot, 0, 1, \ldots, m}$ and two mappings from $C^2$ to $\set{\bot, 0, 1, \ldots, m}$, where $\bot$ encodes that no such trace or transfer exists. The number of quadruples of such mappings is bounded by $(m+2)^{2(k^2 +k)} = 2^{2(k^2+k)\log(m+2)}$.

\ref{lemma_eqproperties_eqwordpar_index}. An equivalence class of $\eqwordpar{m}$ is uniquely characterized by a transition profile and, for every state~$q$, by the $\eqoppar{m}$ equivalence class of the sequence of counter operations encountered along the run starting in $q$. Thus, the class is characterized by a mapping from $Q$ to pairs of a state and an $\eqoppar{m}$ class. Thus, the index of $\eqwordpar{m}$ is bounded by the number of such mappings, i.e., by
\begin{equation*} (n \cdot \idx{\eqoppar{m}})^n = 
2^{\log \left( 
(n  2^{2(k^2+k)\log(m+2)})^n
\right)} = 2^{n\left( \log(n) + 2(k^2+k)\log(m+2) \right)}.
 \tag*{\qed}
\end{equation*}
\end{proof}

Next, we show that we take any infinite word~$x_0 x_1 x_2 \cdots$ with $x_i \in \Sigma^*$ and replace each $x_i$ by an \emph{equivalent} $x_i'$ without changing membership in $L(\aut)$. To capture the evolution of the counters properly with the imprecise equivalence relations~$\eqwordpar{m}$, we require that the $x_i$ and the $x_i'$ are $\eqwordpar{m}$-equivalent for $m$ tending to infinity. Formally, a sequence~$(r_i)_{i \in \nats}$ of natural numbers is a (convergence) rate, if it is weakly increasing and unbounded, i.e., $r_i \le r_{i+1}$ for every $i$ and $\sup_i r_i = \infty$. 

\begin{lemma}
	\label{lemma_swapeqwords}
Let $(x_i)_{i \in \nats}$ and $(x_i')_{i \in \nats}$ be two sequences of words over $\Sigma^*$ and let $(r_i)_{i \in \nats}$ be a rate such that $x_i \eqwordpar{r_i} x_i'$ for all $i$. Then, $x = x_0 x_1 x_2 \cdots \in L(\aut)$ if and only if $x' = x_0' x_1' x_2' \cdots \in L(\aut)$.
\end{lemma}

\begin{proof}
Let $\rho = (q_0, \alpha(0), q_1) (q_1, \alpha(1), q_2)  \cdots$  be the run of $\aut$ on $x$ and let $\rho' = (q_0', \alpha'(0), q_1') (q_1', \alpha'(1), q_2') \cdots$ be the run of $\aut$ on $x'$, i.e., $x = \alpha(0)\alpha(1)\alpha(2) \cdots$ and $x' = \alpha'(0)\alpha'(1)\alpha'(2) \cdots$. Furthermore, let $n_i = \size{x_0 \cdots x_{i-1}}$ and $n_i' = \size{x_0' \cdots x_{i-1}'}$. By definition of $\eqwordpar{m}$, we obtain $q_{n_i} = q_{n_i'}'$ for every $i \ge 0$. Finally, we have $\ell(q_{n_i},x_i) \eqoppar{r_i} \ell(q_{n_i'}',x_i')$ for every $i$, due to  $x_i \eqwordpar{r_i} x_i'$ and $q_{n_i} = q_{n_i'}'$.

We show that $\rho$ contains arbitrarily long $c$-traces if and only if $\rho'$ contains arbitrarily long $c$-traces. Due to Lemma~\ref{lemma_ctraceswitnessunboundedness}, this suffices to show that $\rho$ is accepting if and only if $\rho'$ is accepting. Furthermore, due to symmetry, it suffices to show one direction of the equivalence. Thus, assume $\rho$ contains arbitrarily long $c$-traces and pick $m \in \nats$ arbitrarily. We show the existence of a $c$-trace of length~$m$ contained in $\rho'$. 

To this end, fix a $c$-trace of length~$m$ in $\rho$. We can assume w.l.o.g.\ that the trace is contained in a run infix~$\run(q_{n_i}, x_i \cdots x_{i'})$ (which ends with $q_{n_{i'+1}}$) for some $i \le i'$ with $m \le r_i \le r_{i'}$. Furthermore, we assume w.l.o.g.\ that $i$ ($i'$) is maximal (minimal) with this property for the fixed trace.

If $i = i'$, then the complete $c$-trace is contained in $\run(q_{n_i}, x_i)$, i.e., $\ell(q_{n_i}, x_i)$ has an infix whose flattening has a suffix that is a $c$-trace of length $m \le r_i$. Thus, the first requirement in the definition of $\eqoppar{r_i}$ yields an infix of $\ell(q_{n_i'}', x_i')$ whose flattening has a suffix that is a $c$-trace of length $m$. Thus, $\rho'$ contains a $c$-trace of length~$m$.

If $i < i'$, then the maximality of $i$ and the minimality of $i'$ imply that there are counters $d_0, d_1$ and non-negative numbers~$m_0 + m_1 + m_2 = m$ such that 
\begin{itemize}
	\item the flattening of $\ell(q_{n_i}, x_i)$ has a suffix that is a $d_0$-trace of length~$m_0$,
	\item the flattening of $\ell(q_{n_i+1}, x_{i+1} \cdots x_{i'-1})$ transfers $d_0$ to $d_1$ with $m_1$ increments, and
	\item $\ell(q_{n_i'}, x_{i'})$ has a prefix whose flattening transfers $d_1$ to $c$ with $m_2$ increments.
\end{itemize}
The latter three requirements  in the definition of $\eqoppar{r_i}$ imply the existence of the same transfers and traces in $\ell(q_{n_i'}', x_i')$, $\ell(q_{n_i'+1}', x_{i+1}' \cdots x_{i'-1}')$, and $\ell(q_{n_{i'}'}', x_{i'}')$, respectively. Hence, $\rho'$ contains a $c$-trace of length~$m$.
\end{proof}

The $\eqwordpar{m}$ classes are regular and \textit{trackable} on-the-fly by a finite automaton~$\auttrack_m$ due to $\eqwordpar{m}$ being a congruence.

\begin{lemma}
\label{lemma_eqclasstracker}
There is a deterministic finite automaton~$\auttrack_m$ with set of states~${\Sigma\quotientwordpar{m}}$ such that the run of $\auttrack_m$ on $w \in \Sigma^*$ ends with state~$[w]_{\eqwordpar{m}}$.
\end{lemma}

\begin{proof}
Define $\auttrack_m = (\Sigma\quotientwordpar{m}, \Sigma, \eqclasswordpar{\epsilon}{m}, \delta_{\auttrack_m}, \emptyset)$ where 
$\delta_{\auttrack_m}(\eqclasswordpar{x}{m}, a) = \eqclasswordpar{xa}{m}$,
which is independent of the representative~$x$ and based on the fact that $\eqwordpar{m}$ is a congruence. A straightforward induction over $|w|$ shows that $\auttrack_m$ has the desired properties.
\end{proof}

In particular, every $\eqwordpar{m}$ equivalence class is regular and recognized by the DFA obtained from $\auttrack_m$ by making the class to be recognized the only final state.

For the remainder of this section, we assume $\Sigma = \SigmaI \times \SigmaO$. We denote the projection of $\SigmaI \times \SigmaO$ to $\SigmaI$ by $\proj{\cdot}$, an operation we lift to words and languages over $\SigmaI \times \SigmaO$ in the usual way. Now, for each equivalence relation~$\eqwordpar{m}$ over $(\SigmaI \times \SigmaO)^*$ we define its projection\footnote{The notation~$\eqwordparpro{m}$ should not be understood as denoting equality, but merely as having projected away one bar from $\eqwordpar{m}$.}~$\eqwordparpro{m}$ over $\SigmaI^*$ via $ x\eqwordparpro{m} x'$ if and only if
for all $\eqwordpar{m}$ classes~$S$: $x \in \proj{S}$ if and only if $x' \in \proj{S}$. 

\begin{numremark}
\label{rem_eqprojindex}
$\idx{\eqwordparpro{m}}  \le 2^{\idx{\eqwordpar{m}}}$.	
\end{numremark}

Furthermore, every $\eqwordparpro{m}$ equivalence class is regular: we have
\begin{equation*}
\eqclasswordparpro{x}{m} = \bigcap\nolimits_{S \in (\SigmaI \times \SigmaO)^*\quotientwordpar{m}\colon x \in \proj{S}}
 \proj{S} \cap \bigcap\nolimits_{S \in (\SigmaI \times \SigmaO)^*\quotientwordpar{m} \colon x \notin \proj{S}} \SigmaI^* \setminus \proj{S},
 \label{eq_projrelcharac}
\end{equation*}
where each projection~$\proj{S}$ and each complemented projection~$\SigmaI^* \setminus \proj{S}$ is recognized by a DFA of size~$2^{\idx{\eqwordpar{m}}}$. Thus, $\eqclasswordparpro{x}{m}$ is recognized by a DFA of size~$ 2^{\idx{\eqwordpar{m}}^2}$. In particular, we have the following bound that will be applied in the next subsection.

\begin{numremark}
\label{remark_infclasses}
Let $x$ be in a finite equivalence class of $\eqwordparpro{0}$. Then, we have \newline $\size{x} < 2^{ 2^{2n(\log(n) + 2(k^2+k))} }$.	
\end{numremark}

\subsection{A Game on Equivalence Classes}
\label{subsec_equivgame}
In this section, we show that the winner of a delay game with max-regular winning condition does not depend on the exact delay function under consideration, as long as it is unbounded and $f(0)$ is \emph{large enough}. Note that this is true for the game analyzed in the previous section: Player~$O$ wins for every unbounded delay function. 

\begin{theorem}
\label{thm_anyunboundedsufficient}
Let $\aut$ be a max-automaton with $n$ states and $k$ counters and let $d =  2^{ 2^{2n(\log(n) + 2(k^2 + k))} }$. The following are equivalent:
\begin{enumerate}
	\item Player~$O$ wins $\delaygame{L(\aut)}$ for some $f$.
	\item Player~$O$ wins $\delaygame{L(\aut)}$ for every unbounded~$f$ with $f(0) \ge 2d$.
\end{enumerate}
\end{theorem} 

This result is proven by defining a delay-free game~$\game(\aut)$ where Player~$I$ picks equivalence classes of $\eqwordparpro{m}$ for increasing $m$ and Player~$O$ constructs a run of $\aut$ on a word over $\SigmaI \times \SigmaO$ that is compatible with the choices of Player~$I$. Furthermore, Player~$I$ is always  one move ahead to account for the delay.

Fix $\aut = (Q, C, \SigmaI \times \SigmaO, q_I, \delta, \ell, \phi)$ with $\size{Q} = n$ and $\size{C} = k$. We define the game~$\game(\aut)$ between Player~$I$ and Player~$O$ played in rounds~$i = 0, 1, 2, \ldots$ as follows:
In round~$0$, Player~$I$ picks natural numbers~$r_0, r_1$ and picks infinite equivalence classes~$\eqclasswordparpro{x_0}{r_0}$ and $\eqclasswordparpro{x_1}{r_1}$. Then, Player~$O$ picks an equivalence class~$\eqclasswordpar{{ x_0 \choose y_0 }}{r_0}$. Note that this choice is independent of the representative~$x_0$. Now, consider round~$i>0$: Player~$I$ picks~$r_{i+1} \in \nats$ and an infinite equivalence class~$\eqclasswordparpro{x_{i+1}}{r_{i+1}}$. Afterwards, Player~$O$ picks an equivalence class~$\eqclasswordpar{{ x_{i} \choose y_{i} }}{r_{i}}$, whose choice is again independent of the representative~$x_i$. 

Thus, the players produce a play \[\eqclasswordparpro{ x_0 }{r_0}\,
\eqclasswordpar{{ x_0 \choose y_0 }}{r_0}\,
\eqclasswordparpro{ x_1 }{r_1}\,
\eqclasswordpar{{ x_1 \choose y_1 }}{r_1}\,
\eqclasswordparpro{ x_2 }{r_2}\,
\eqclasswordpar{{ x_2 \choose y_2 }}{r_2}\,
\cdots.\]
Player~$O$ wins, if $(r_i)_{i \in \nats}$ is not a rate or if ${ x_0 \choose y_0 }{ x_1 \choose y_1 }{ x_2 \choose y_2 } \cdots \in L(\aut)$. Otherwise, i.e., if $(r_i)_{i \in \nats}$ is a rate and ${ x_0 \choose y_0 }{ x_1 \choose y_1 }{ x_2 \choose y_2 } \cdots \notin L(\aut)$, Player~$I$ wins. By Lemma~\ref{lemma_swapeqwords}, winning does not depend on the choice of representatives $x_i$ and $y_i$. Strategies and winning strategies for $\game(\aut)$ are defined as expected, taking into account that Player~$I$ is always one equivalence class ahead.

The following lemma about the relation between $\delaygame{L(\aut)}$ and $\game(\aut)$ implies Theorem~\ref{thm_anyunboundedsufficient}. 

\begin{lemma}
\label{lemma_gameequivalences}
The following are equivalent:
\begin{enumerate}
	
	\item\label{lemma_gameequivalences_somef}
	Player~$O$ wins $\delaygame{L(\aut)}$ for some $f$.
	
	\item\label{lemma_gameequivalences_allf} 
	Player~$O$ wins $\delaygame{L(\aut)}$ for every unbounded~$f$ with $f(0) \ge 2d$.
	
	\item\label{lemma_gameequivalences_eqgame} 
	Player~$O$ wins $\game(\aut)$.

\end{enumerate}
\end{lemma} 

\begin{proof}
It suffices to show that \ref{lemma_gameequivalences_somef}.\ implies \ref{lemma_gameequivalences_eqgame}.\ and that \ref{lemma_gameequivalences_eqgame}.\ implies \ref{lemma_gameequivalences_allf}., as \ref{lemma_gameequivalences_allf}.\ implies \ref{lemma_gameequivalences_somef}.\ is trivially true. For the sake of readability, we will write $\Gamma$ instead of $\delaygame{L(\aut)}$, as long as $f$ is clear from context. Similarly, we will write $\game$ instead of $\game(\aut)$.

Let Player~$O$ win $\delaygame{L(\aut)}$ for some $f$, say with winning strategy~$\stratO$. We construct a winning strategy~$\stratO'$ for her in $\game$ by simulating a play in $\Gamma$ that is consistent with $\stratO$.

In round~$0$ of $\game$, Player~$I$ picks $r_0$,$r_1$, $\eqclasswordparpro{x_0}{r_0}$, and $\eqclasswordparpro{x_1}{r_1}$. As both equivalence classes are infinite, we can assume without loss of generality $\size{x_0} \ge f(0)$ and $\size{x_1} \ge \sum_{j=1}^{\size{x_0}-1}f(j)$. Now, assume Player~$I$ picks in $\Gamma$ the prefix of $x_0x_1$ of length~$\sum_{j=0}^{\size{x_0}-1}f(j)$ during the first $\size{x_0}$ rounds. Let $y_0$ of length $\size{x_0}$ be the answer of Player~$O$ to these choices determined by the winning strategy~$\stratO$. We define $\stratO'$ such that it picks $\eqclasswordpar{{x_0 \choose y_0}}{r_0}$ as answer to Player~$I$ picking $r_0$,$r_1$, $\eqclasswordparpro{x_0}{r_0}$, and $\eqclasswordparpro{x_1}{r_1}$ in round~$0$.

Now, we are in the following situation for $i=1$: in $\game$, Player~$I$ has picked natural numbers~$r_0, \ldots, r_i$ and $\eqclasswordparpro{x_0}{r_0}, \ldots, \eqclasswordparpro{x_i}{r_i}$ with $\size{x_0} \ge f(0)$, $\size{x_1} \ge \sum_{j=1}^{\size{x_0}-1}f(j)$, and $\size{x_{i'}} \ge \sum_{j=0}^{\size{x_{i'-1}}-1}f(\size{x_0 \cdots x_{i'-1}} + j)$ for every $i'$ with $1 < i' \le i$ (which is empty for $i=1$). Player~$O$ has picked $\eqclasswordpar{{x_0 \choose y_0}}{r_0}, \ldots, \eqclasswordpar{{x_{i-1} \choose y_{i-1}}}{r_{i-1}}$. Further, in $\Gamma$, Player~$I$ has picked the prefix of $x_0 \cdots x_i$ of length~$\sum_{j=0}^{\size{x_0\cdots x_{i-1}}-1}f(j)$ during the first $\size{x_0 \cdots x_{i-1}}$ rounds, which was answered by Player~$O$ according to $\stratO$ by picking $y_0 \cdots y_{i-1}$.

In this situation, it is Player~$I$'s turn in $\game$, i.e., he picks $r_{i+1}$ and $\eqclasswordparpro{x_{i+1}}{r_{i+1}}$. Again, as the  class is infinite, we can assume $\size{x_{i+1}} \ge \sum_{j=0}^{\size{x_{i}}-1}f(\size{x_0 \cdots x_{i}} + j)$.
Thus, we continue the play in $\Gamma$ by letting Player~$I$ pick letters such that he has picked the prefix of $x_0 \cdots x_{i+1}$ of length~$\sum_{j=0}^{\size{x_0\cdots x_{i}}-1}f(j)$ during the first $\size{x_0 \cdots x_{i}}$ rounds. Again, this is answered by Player~$I$ by picking $y_0 \cdots y_i$ such that $\size{y_{i'}} = \size{x_{i'}}$ according to $\stratO$. Now, we define $\stratO'$ such that it picks $\eqclasswordpar{{x_i \choose y_i}}{r_i}$ as next move. Thus, we are in the situation described above for $i+1$.

Let $w' = \eqclasswordparpro{ x_0 }{r_0}\,
\eqclasswordpar{{ x_0 \choose y_0 }}{r_0}\,
\eqclasswordparpro{ x_1 }{r_1}\,
\eqclasswordpar{{ x_1 \choose y_1 }}{r_1}\,
\eqclasswordparpro{ x_2 }{r_2}\,
\eqclasswordpar{{ x_2 \choose y_2 }}{r_2}\,
\cdots$ be a play in $\game$ that is consistent with $\stratO'$. Consider the outcome~$w = {x_0 \choose y_0 }{x_1 \choose y_1 }{x_2 \choose y_2 }\cdots $ of the play in $\Gamma$ constructed during the simulation. It is consistent with $\stratO$, hence $w \in L(\aut)$. Accordingly, Player~$O$ wins the play~$w'$. Thus, $\stratO'$ is indeed a winning strategy for Player~$O$ in $\game$.


Now, consider the second implication to be proven: assume Player~$O$ has a winning strategy~$\stratO'$ for $\game$ and let $f$ be an arbitrary unbounded delay function with $f(0) \ge 2d$. We construct a winning strategy~$\stratO$ for Player~$O$ in $\Gamma$ by simulating a play of $\game$.

To this end, we define a strictly increasing auxiliary rate~$(d_i)_{i \in \nats}$ recursively as follows: let $d_0$ be minimal with the property that every word of length at least $d_0$ is in some infinite equivalence class of $\eqwordparpro{0}$, i.e., $d_0 \le d = 2^{{  2^{2n(\log n + 2(k^2 + k))} }}$ due to Remark~\ref{remark_infclasses}. Now, we define $d_{i+1}$ to be the minimal integer strictly greater than $d_i$ such that every word of length at least $d_{i+1}$ is in some infinite equivalence class of $\eqwordparpro{i+1}$.

Let Player~$I$ pick $x_0 x_1$ of length~$f(0) \ge 2\cdot d_0$ in round~$0$ of $\Gamma$ (the exact decomposition into $x_0$ and $x_1$ is irrelevant, we just use it to keep the notation consistent). Now, decompose $x_0 x_1 = x_0' x_1' \beta_1$ such that $\size{x_0'} = \size{x_1'} = d_0$. We simulate these moves by letting Player~$I$ pick $r_0 = r_1 = 0$, $\eqclasswordparpro{x_0'}{r_0}$, and $\eqclasswordparpro{x_1'}{r_1}$ in round~$0$ of $\game$, which are legal moves by the choice of $d_0$. 

Thus, we are in the following situation for $i=1$: in $\Gamma$, Player~$I$ has picked $x_0 \cdots x_i$ and Player~$O$ has picked $y_0 \cdots y_{i-2}$. Furthermore, in $\game$, Player~$I$ has picked $\eqclasswordparpro{x_0'}{r_0}, \ldots, \eqclasswordparpro{x_i'}{r_i}$ and there is a buffer~$\beta_i \in \SigmaI^*$ such that $x_0 \cdots x_i = x_0' \cdots x_i' \beta_i$. Finally, Player~$O$ has picked $\eqclasswordpar{{x_0' \choose y_0}}{r_0} \cdots \eqclasswordpar{{x_{i-2}' \choose y_{i-2}}}{r_{i-2}}$.

In this situation, it is Player~$O$'s turn and $\stratO'$ returns a class~$\eqclasswordpar{{x_{i-1}' \choose y_{i-1}}}{r_{i-1}}$. Thus, we define $\stratO$ such that it picks $y_{i-1}$ during the next rounds, in which Player~$I$ picks letters forming $x_{i+1}$ satisfying $\size{x_{i+1}} \ge \size{y_{i-1}}$. We consider two cases to simulate these in $\game$:
\begin{enumerate}
	
	\item If $\size{\beta_ix_{i+1}} \ge 2d_{r_i+1} - d_{r_i}$, then Player~$I$ picks $r_{i+1} = r_i+1$ and $\eqclasswordparpro{x_{i+1}'}{r_{i+1}}$, where $x_{i+1}'$ is the prefix of $\beta_ix_{i+1}$ of length~$d_{r_{i+1}}$. This is an infinite equivalence class by the choice of $d_{r_{i+1}}$. The remaining suffix of $\beta_ix_{i+1}$ is stored in the buffer~$\beta_{i+1}$, i.e., we have $\beta_ix_{i+1} = x_{i+1}' \beta_{i+1}$.
	 
	\item Now, consider the case $\size{\beta_ix_{i+1}} < 2d_{r_i+1} - d_{r_i}$: we show $\size{\beta_ix_{i+1}} \ge d_{r_i}$. Then, Player~$I$ picks $r_{i+1} = r_i$ and $\eqclasswordparpro{x_{i+1}'}{r_{i+1}}$, where $x_{i+1}'$ is the prefix of $\beta_ix_{i+1}$ of length~$d_{r_{i+1}} = d_{r_{i}}$, which is again an infinite equivalence class by the choice of $d_{r_{i+1}}$. The remaining suffix of $\beta_ix_{i+1}$ is stored in the buffer~$\beta_{i+1}$, i.e., we have $\beta_ix_{i+1} = x_{i+1}' \beta_{i+1}$.

To show 	$\size{\beta_ix_{i+1}} \ge d_{r_i}$, we again consider two cases: if $r_{i-1} = r_i$, then we have 
\[
\size{\beta_ix_{i+1}} \ge \size{x_{i+1}} \ge \size{y_{i-1}} = \size{x_{i-1}'} = d_{r_{i-1}} = d_{r_i}. 
\]
On the other hand, if $r_{i-1} < r_i$, which implies $r_{i-1} +1  = r_i$, as we are in the second case above, then we have $
\size{\beta_{i-1}x_{i}} \ge 2d_{r_{i-1}+1} - d_{r_{i-1}}$
and $x_i'$ is the prefix of length $d_{r_i} = d_{r_{i-1}+1}$ of $\beta_{i-1}x_{i}$, which implies $\size{\beta_i} \ge d_{r_{i-1}+1} - d_{r_{i-1}} $, as it is the remaining suffix of $\beta_{i-1}x_{i}$. Finally, we have $\size{x_{i+1}} \ge \size{y_{i-1}} = \size{ x_{i-1}'} = d_{r_{i-1}}$.
Altogether, we obtain
\[
\size{ \beta_ix_{i+1} } \ge (d_{r_{i-1}+1} - d_{r_{i-1}}) + d_{r_{i-1}} = d_{r_{i-1}+1} = d_{r_i}.
\] 
\end{enumerate}
In both cases, we are back in the situation described above for $i+1$.

Let $w = {x_0 x_1 x_2 \cdots \choose y_0 y_1 y_2 \cdots}$ be the outcome of a play in $\Gamma$ that is consistent with $\stratO$. The play~$\eqclasswordparpro{ x_0' }{r_0}\,
\eqclasswordpar{{ x_0' \choose y_0 }}{r_0}\,
\eqclasswordparpro{ x_1' }{r_1}\,
\eqclasswordpar{{ x_1' \choose y_1 }}{r_1}\,
\eqclasswordparpro{ x_2' }{r_2}\,
\eqclasswordpar{{ x_2' \choose y_2 }}{r_2}\,
\cdots $ in $\game$ constructed during the simulation is consistent with $\stratO'$. As $f$ is unbounded, $(r_i)_{i \in \nats}$ is unbounded as well and thus a rate. Hence, we conclude ${x_0' \choose y_0}{x_1' \choose y_1}{x_2' \choose y_2} \in L(\aut)$, as $\stratO'$ is a winning strategy. Also, a straightforward induction shows $x_0 x_1 x_2 \cdots = x_0' x_1' x_2 ' \cdots$. Thus, $w \in L(\aut)$, i.e., $\stratO$ is a winning strategy for Player~$O$ in $\Gamma$.
\end{proof}

%
%
%

\section{Solving Max-regular Delay games with Unbounded Lookahead}
\label{sec_decidability}
Unlike for $\omega$-regular delay games, bounded lookahead is not always sufficient for Player~$O$ to win a max-regular delay game. Hence, determining the winner with respect to arbitrary delay functions is not equivalent to determining the winner with respect to bounded delay functions, which is known to be decidable~\cite{Zimmermann15}. We refer to the former problem as \myquot{solving max-regular delay games}. In this concluding section, we discuss some obstacles one has to overcome in order to extend the decidability result for bounded lookahead to unbounded one. Furthermore, we give straightforward lower bounds on the complexity. 

Proving upper bounds, e.g., decidability of determining the winner of max-regular delay games with respect to arbitrary delay functions, is complicated by the need for unbounded lookahead. All known decidability results are for the case of bounded lookahead~\cite{HoltmannKaiserThomas12,KleinZimmermann15,Zimmermann15}. In particular, decidability of max-regular delay games with respect to bounded lookahead~\cite{Zimmermann15} is based on a game similar to the game~$\game(\aut)$ presented in Section~\ref{sec_unbounded}, but where only equivalence classes of $\eqwordparpro{1}$ and $\eqwordpar{1}$ are picked by the players. This results in a finite delay-free game with max-regular winning condition, which is effectively solvable~\cite{Bojanczyk14}. Correctness follows from the fact that the error introduced by the imprecise equivalence relation~$\eqwordpar{1}$ is bounded, if the lookahead is bounded. As we are only interested in (un)boundedness, this error is negligible. 

However, for unbounded lookahead, the error is unbounded as well. In particular, the example presented in Section~\ref{sec_bounded} shows that bounded counters might grow arbitrarily large during different plays: the winning condition $L$ described in the proof of Theorem~$\ref{thm_unboundednecessary}$ is recognized by a max-automaton with four counters: $c_i$ counts the length of input blocks and is reset at every $\sepp$, $c_o'$ is incremented during prefixes of possible output blocks and reset at the end of such a block. Furthermore, the value of $c_o'$ is copied to $c_o$ every time the requirement on the first and last letter of an output block is met. Finally, a counter~$c_\sepp$ counts the number of $\sepp$'s in the word. The acceptance condition of the automaton recognizing $L$ is given by the formula
\[
\text{\myquot{$\limsup \rho_{c_\sepp} < \infty$}}\,\, \vee \,\, 
\text{\myquot{$\limsup \rho_{c_i} < \infty$}}\,\, \vee \,\,
\text{\myquot{$\limsup \rho_{c_o} = \infty$}}
.\]
As already argued, Player~$O$ has a winning strategy for $\delaygame{L}$, provided $f$ is unbounded. However,  she does not have a strategy that bounds the counters $c_\sepp$ and $c_i$ to some fixed value among all consistent plays that are won due to $c_\sepp$ or $c_i$ being bounded: for example, Player~$I$ can pick any finite number of $\sepp$'s and then stop doing so. This implies that $c_\sepp$ is bounded, but with an arbitrarily large value among different plays. The lack of such a uniform bound in itself is not surprising, but entails that one has to deal with arbitrarily large counter values when trying to extend the approach described above for the setting with bounded lookahead. In particular, it is not enough to replace $\eqwordparpro{1}$ and $\eqwordpar{1}$ by $\eqwordparpro{m}$ and $\eqwordpar{m}$ for some fixed $m$ that only depends on the winning condition.

Two other possible approaches follow from the results mentioned in this paper: first, one could show that $\game(\aut)$ can be solved effectively. However, the game is of infinite size and not in one of the classes of effectively solvable infinite games, e.g., pushdown games. Second, one can pick any unbounded delay function with $f(0)$ large enough and solve $\delaygame{L(\aut)}$, as winning with respect to one such function is equivalent to winning with respect to all of them. However, $\delaygame{L(\aut)}$ is again infinite and not in in one of the classes of effectively solvable infinite games.

Finally, there is a class of winning conditions for which solving delay games is indeed known to be undecidable, namely (very restricted fragments of) $\omega$-context-free conditions~\cite{FridmanLoedingZimmermann11}. However, this result is based on the language~$\set{a^nb^n \mid n \in \nats}$ being context-free, which suffices to encode two-counter machines. As max-automata have no mechanism to compare arbitrarily large numbers exactly, this simple encoding of two-counter machines cannot be captured in a delay game with max-regular winning condition. 

This can be overcome by allowing quantification over arbitrary sets: recently, and after being an open problem for more than a decade, satisfiability of MSO$+$U over infinite words was shown to be undecidable~\cite{BPT15} by capturing termination of two-counter machines by MSO$+$U formulas based on a specially tailored encoding. However, the resulting formulas have six alternations between existential and universal set quantifiers and then a block of (negated) unbounding quantifiers. To adapt this proof to show undecidability of max-regular delay games with respect to arbitrary delay functions, one has to replace the set quantifiers by the interaction between the players, which seems unlikely to achieve.

On the other hand, one can prove some straightforward lower bounds. As usual, solving delay games with max-regular winning conditions (given by max-automata) is at least as hard as solving the universality problem for max-automata: given such an automaton~$\aut$ over some alphabet~$\Sigma$, we change the alphabet to $\Sigma \times \Sigma$ by replacing each letter~$a$ on a transition by the letter~${a \choose a}$. Call the resulting automaton~$\aut'$. The game~$\delaygame{\aut'}$ is won by Player~$O$ if and only if $L(\aut)$ is universal, independently of $f$: if $L(\aut)$ is not universal, when Player~$I$ can produce some $\alpha \notin L(\aut)$ and thereby win; if it is indeed universal, then Player~$O$ can mimic the choices of Player~$I$ and thereby win. 

\begin{proposition}
Solving max-regular delay games is at least as hard as solving the universality problem for max-automata.	
\end{proposition}

The best known lower bound on the universality problem for max-automata is $\pspace$-hardness, which stems from max-automata being closed under complementation and the emptiness problem being $\pspace$-hard~\cite{BojanczykTorunczyk09}. The exact complexity of the emptiness problem for max-automata is, to the best of our knowledge, an open problem.

Another lower bound is obtained by considering delay games with weaker winning conditions: solving delay games with winning conditions recognized by deterministic safety automata is $\exptime$-complete~\cite{KleinZimmermann15}. Such automata can be transformed into max-automata without increasing the number of states: turn the non-safe states into sinks and increment a designated counter~$c$ on every transition not leading into a non-safe state. Then, the max-automaton with acceptance condition~\myquot{$\limsup \rho_c = \infty$} recognizes the same language as the original safety automaton. Hence, we obtain the following lower bound.

\begin{theorem}
	Solving max-regular delay games is $\exptime$-hard. 
\end{theorem}

This lower bound is oblivious to the intricate acceptance condition of max-automata and relies solely on the transition structure. This is in line with results for $\omega$-regular games: solving delay games with winning conditions given by deterministic parity automata is in $\exptime$, i.e., it matches the lower bound for the special case of safety. It is open whether moving to more concise acceptance conditions for deterministic $\omega$-automata, e.g., Rabin, Streett, and Muller, increases the complexity. These results would directly transfer to max-automata as well. Another aspect that is not exploited by this reduction is the  unbounded lookahead: the safety delay game is always winnable with bounded lookahead. We are currently investigating whether these two aspects can be exploited to improve the bounds.


\bibliographystyle{plain}
\bibliography{biblio}


\end{document}